\documentclass{article}

\usepackage{arxiv}

\usepackage[utf8]{inputenc} 
\usepackage[T1]{fontenc}    
\usepackage{hyperref}       
\usepackage{url}            
\usepackage{booktabs}       
\usepackage{amsfonts}       
\usepackage{nicefrac}       
\usepackage{microtype}      
\usepackage{lipsum}
\usepackage{fancyhdr}

\usepackage{graphicx}          
\usepackage{amsmath} 
\usepackage{amssymb}  

\usepackage{amsthm}
\usepackage{color,soul}
\usepackage{cite}
\usepackage{tabularx}
\usepackage{subcaption}
\usepackage{enumitem}
\usepackage{wrapfig}

\allowdisplaybreaks
\newtheorem{thmm}{Theorem}

\newtheorem{remm}{Remark}
\newtheorem{assm}{Assumption}
\newtheorem{deff}{Definition}
\newtheorem{Lemm}{Lemma}

\newtheorem{DR}{Design Requirement}

\title{Security of Distributed Parameter Cyber-Physical Systems: Cyber-Attack Detection in Linear Parabolic PDEs}

\author{Tanushree~Roy,~and~Satadru~Dey
\thanks{T. Roy and S. Dey are with the Department of Mechanical Engineering,
        The Pennsylvania State University, University Park, Pennsylvania 16802, USA.
        {\tt\small \{tbr5281,skd5685\}@psu.edu}.}%
}

\begin{document}
\maketitle

\begin{abstract}
Security of Distributed Parameter Cyber-Physical Systems (DPCPSs) is of critical importance in the face of cyber-attack threats. Although security aspects of Cyber-Physical Systems (CPSs) modelled by Ordinary differential Equations (ODEs) have been extensively explored during the past decade, security of DPCPSs has not received its due attention despite its safety-critical nature. In this work, we explore the security aspects of DPCPSs from a system theoretic viewpoint. Specifically, we focus on DPCPSs modelled by linear parabolic Partial Differential Equations (PDEs) subject to cyber-attacks in actuation channel. First, we explore the detectability of such attacks and derive conditions for stealthy attacks. Next, we develop a design framework for cyber-attack detection algorithms based on output injection observers. Such attack detection algorithms explicitly consider stability, robustness and attack sensitivity in their design. Finally, theoretical analysis and simulation studies are performed to illustrate the effectiveness of the proposed approach.
\end{abstract}


%

\section{INTRODUCTION}

Security of Cyber-Physical Systems (CPSs) with respect to cyber-attacks is essential for safe and reliable operation. 
Although secure operation and control of CPSs modelled by Ordinary differential Equations (ODEs) has been an active research area during the past decade (see \cite{teixeira2015secure,pasqualetti2013attack}, and the references therein), such framework has been significantly under-explored for Distributed Parameter Cyber-Physical Systems (DPCPSs) modelled by Partial Differential Equations (PDEs) \cite{lighthill1955kinematic,parashar2004continuum}. Here, we present a framework to analyze and design cyber-attack detection algorithms for a class of DPCPSs, namely, CPSs modelled by linear parabolic PDEs.

Safety of Distributed Parameter Systems (DPSs) with respect to physical faults has received moderate attention in existing literature. For example, the following are some notable works that designed model-based physical fault diagnosis algorithms for DPSs \cite{Armaou2008,Demetriou2012,El-Farra2007,Ghantasala2009,Yao2014,Baniamerian2012,Ferrari2013,Ferdowsi2014,Demetriou2002,Demetriou2007,deutscher2016fault}. However, as discussed in \cite{teixeira2015secure}, physical fault diagnostic algorithms might not be suitable for attacks as the later possesses intelligence and can be launched in a coordinated manner to satisfy specific malicious objectives \cite{teixeira2015secure}. 
This motivates the need for separate methodologies specific to cyber-attack analysis and detection \cite{teixeira2015secure}. In this context, the objective of our work is to develop a mathematical framework that enable analysis and design of cyber-attack detection algorithms for DPCPSs.

There are very few works in existing literature that focus on cyber-attack related issues in DPCPSs. For example, a PDE-based framework is used to analyze the impact of automotive cyber-attacks on highways in \cite{ghanavati2018analysis}. In \cite{reilly2016creating}, a PDE model is utilized to evaluate the capability of attackers to launch attacks on the freeway traffic control system. However, these approaches do not discuss design of cyber-attack detection algorithms. In \cite{demetriou2018detection}, cyber-attack detection algorithms have been designed for a spatially distributed system. However, the problem setting in \cite{demetriou2018detection} assumes the existence of multiple in-domain sensor-actuator pairs. Such assumption may not be true for most of the existing distributed parameter systems where only boundary sensing is available. Furthermore, presence of uncertainties make the design of model-based detection algorithms significantly challenging. It is desirable that any model-based detection algorithms should be able to reject the effect of uncertainties (minimizing false alarms) at the same time being highly sensitive to cyber-attacks (minimizing miss-detection). However, there exists a fundamental trade-off between these two objectives. Most of the aforementioned detection algorithms do not explicitly consider such trade-off in their design. In summary, a framework for analyzing and designing cyber-attack detection algorithms which (i) focuses on DPCPSs with boundary sensing, and (ii) explicitly considers the trade-off between robustness and attack sensitivity, remains an open problem. In our preliminary work \cite{RoyDey2020}, we have studied a similar approach for coupled hyperbolic PDEs with application to traffic networks. However, such framework has not been explored for parabolic PDE systems to the best of our knowledge.

In context of the above review, the pivotal contribution of this paper is in \textit{the development of a cyber-attack analysis and detection framework for boundary-sensed linear parabolic PDEs considering the trade-off between robustness and attack sensitivity}. Specifically:
\begin{enumerate}[label={\roman{enumi})}]
\item We consider a class of boundary-measured linear parabolic PDEs where cyber-attacks can potentially compromise the actuation channels. 
\item We explore the fundamental limitations of model-based algorithms by investigating the conditions of launching a \textit{stealthy attack}.
\item We develop a PDE output injection observer based framework to design model-based algorithms. Specifically, the framework incorporates several design requirements such as stability, robustness to uncertainties, and sensitivity to cyber-attacks. Furthermore, the framework specifically detects cyber-attacks allowing for a provision to distinguish them from physical faults. 
\end{enumerate}

The rest of the paper is organized as follows. Section II details the problem setup. Section III discusses the fundamental shortcomings of the model-based algorithms. Section IV details the design of cyber-attack detection algorithms. Section V presents simulation cases studies that illustrate the performance of the proposed algorithm. Finally, Section VI concludes the work. We have used the following notation, inequalities and lemmas in this paper.
\\\\
\noindent \textbf{Notation:} We use the following notation: $u_t = {\partial u}/{\partial t}$, $u_x = {\partial u}/{\partial x}$, $u_{xx} = {\partial^2 u}/{\partial x^2}$; $\left| X\right|$ denotes the absolute value of $X$; $\left\|u(.) \right\|_{C(D)}$ denotes the supremum norm on the space of continuous functions defined as follows $\left\|u(.) \right\|_{C(D)}:=\sup_{x\in D}|u(x)|$; $\left\|u(.) \right\|$ denotes the spatial $L^2$ norm given as $\left\|u(.) \right\| := \sqrt{\int_0^1 u^2(x) dx}$; $\left\|u(.) \right\|_\mathcal{H}$ denotes the following $\left\|u(.) \right\|_\mathcal{H} := \sqrt{\left|u(1) \right|^2 + \left\|u(.) \right\|^2 + \left\|u_x(.) \right\|^2}$.

\section{PROBLEM SETUP}
We consider a class of DPCPSs where the plant resides in the physical layer and the controller resides in the cyber layer. The plant operation is enabled by a communication network which exchanges information between the plant and the controller. The security concerns of such DPCPSs arise from the vulnerability of the communication channels. Specifically, we consider the cyber-attacks where an adversary can potentially access and modify the information in the actuation channels. Actuator attacks can impact system functionality significantly and deserves attention while designing secure CPSs \cite{fawzi2014secure,jin2017adaptive,urbina2016limiting,guo2017vcids}. These actuator attacks can be either directly in the actuator (for example, by manipulating the actuator software) or by injecting false data into the actuator command \cite{fawzi2014secure}.

The physical plant is modelled by the following class of linear PDEs in time $t \in [0,\infty)$ and space $x \in [0,1]$:
\begin{equation}
    u_t(x,t) = u_{xx}(x,t) + \alpha u(x,t) + D(x)q(t) + D_a(x)\delta(t), \label{pde}
\end{equation}
with boundary conditions
\begin{equation}
    u_x(0,t) = 0, \quad u_x(1,t) = 0,  \label{pde-bc}
\end{equation}
and initial condition $u(x,0) = \phi(x)$ along with the measured boundary state
\begin{equation}
    y(t)=u(1,t),  \label{pde-output}
\end{equation}
where $u(x,t)$ is the distributed state, $\alpha$ is the reaction coefficient, $q(t)$ is the nominal in-domain input applied through the actuation distribution function $D(x)$, $\delta(t)$ is the exogenous actuation attack component applied through the distribution function $D_a(x)\subseteq D(x)$. The assumption of $D_a(x)$ being a subset of $D(x)$ is due to the fact that the actuation can be partially compromised. We assume that the distributed parameter plant exhibits stable behavior under the control system, and hence, $\alpha<0$. The PDE \eqref{pde}-\eqref{pde-bc} represents systems involving diffusion (captured by $u_{xx}$) and reaction (captured by $\alpha u$), for example, thermal systems \cite{muratori2012spatially,reentry,cancer}.  

We also consider a class of attack detectors to detect the occurrences of the cyber-attacks. The attack detection algorithm takes the following form:
\begin{align}\label{ada}
    \mathcal{A}(.): \{\mathcal{M}_N,y(t) \forall t\geq 0\} \mapsto r(t),
\end{align}
where $\mathcal{M}_N$ indicates the nominal model \eqref{pde}-\eqref{pde-bc} with $\delta(t)=0$ and unknown initial condition $\phi(x)$, and $r(t)$ is a residual signal which is used for cyber-attack detection based on the following logic.
\begin{align}\label{residual}
    &r(t) = 0 \Rightarrow \text{No attack},r(t) \neq 0 \Rightarrow \text{Attack occurrence}.
\end{align}

Under this problem setup, our objective is to understand and quantify the fundamental performance limitation of $\mathcal{A}(.)$ and develop a framework for designing $\mathcal{A}(.)$.

\section{FUNDAMENTAL LIMITATION OF ATTACK DETECTOR ALGORITHMS}
In this section, we analyze the fundamental limitation of $\mathcal{A}(.)$ in \eqref{ada}. Specifically, we derive the conditions under which an attack will be undetectable by $\mathcal{A}(.)$, irrespective of their design method. We start with the definition of stealthy or undetectable attacks adopted from \cite{pasqualetti2013attack}.

\begin{deff}[Stealthy Attacks]
Denoting the output of the PDE \eqref{pde}-\eqref{pde-output} as $y(t,\phi(x),\delta(t),q(t))$, a non-zero attack $\delta(t)$ is stealthy or undetectable if and only if $y(t,\phi_1(x),\delta(t),q(t))=y(t,\phi_2(x),0,q(t)),t\geqslant 0$ for some initial conditions $\phi_1(x)$ and $\phi_2(x)$ satisfying well-posedness of \eqref{pde}-\eqref{pde-output}.
\end{deff}

With this definition, our objective is to analyze the existence and uniqueness of a stealthy attack in the PDE system \eqref{pde}-\eqref{pde-output}. Without loss of generality, we use the following setting in this analysis: $q(t)=0$, $\phi_1(x) = \phi(x) \neq 0$, and $\phi_2(x) = 0$, under which the stealthiness definition boils down to.:
\begin{equation}\label{stealthiness}
    y(t,\phi(x),\delta(t),0)=y(t,0,0,0)=0,\forall t\geqslant 0.
\end{equation}
Generally, the forward problem for the PDE \eqref{pde}-\eqref{pde-output} entails solution of the state $u(x,t)$ knowing the system model, initial condition $\phi(x)$ and input $\delta(t)$. On the other hand, an inverse problem for \eqref{pde}-\eqref{pde-output} boils down to solving the input $\delta(t)$ while the state $u(x,t)$ (or, in our case the measurement $y(t)=u(1,t)$) is known. Accordingly, the question about existence and uniqueness of stealthy attacks can alternatively be framed as the existence of an unique solution to the inverse problem for \eqref{pde}-\eqref{pde-output}. 

In the next, subsection we present a generalized solution for the forward problem of the  PDE \eqref{pde}-\eqref{pde-output}.



\subsection{Generalized solution of the forward problem}

In order to formulate the generalized solution, we will introduce some important definitions and assumptions that will be invoked in our analysis.

\begin{deff}[Sobolev space $H^1(\Omega)$]
The Sobolev space $H^1(\Omega)$ is a Hilbert space with a norm induced by the following inner product:
\begin{align*}
    \langle h,g \rangle_{H^1} = \int_\Omega (h(x)g(x)+h'(x) g'(x)) dx
\end{align*}
\end{deff}

\begin{deff}[Spaces $H^1_0$ and $H^{-1}$]
The Hilbert space $H^1_0(\Omega)$ is a subspace of $H^1(\Omega)$ and is defined as 
\begin{align}\nonumber
 H^1_0(\Omega)=&\{h\in H^1(\Omega): \exists \text{ a sequence of } h_n \to h \text{ in }H^1(\Omega) \text{ norm with } h_n \in C_0^\infty(\Omega)\},   
\end{align}
where $C_0^\infty$ is the space of compactly supported smooth functions.

On the other hand, the space $H^{-1}(\Omega)$ is the dual space of $H^1_0(\Omega)$ i.e. it consists of all linear functionals of $H^1_0(\Omega)$. 

The duality between $H^{-1}(\Omega)$ and $H^1_0(\Omega)$  is given by ${\langle.,. \rangle: H^{-1}(\Omega)\times H^1_0(\Omega) \to \mathbb{R}} $.
\end{deff}

\begin{deff}[Space of $ L^2(I;X)$ ]
An element  $h(.,t) \in L^2(I;X(\Omega))$ is an element in $X(\Omega)$ for each $t \in I$, with property that $\|h(.,t)\|_{X}$ is a measurable function on $I$ and $\int_I \|h(.,t)\|^2_{X}\, dt<\infty$. 
\end{deff}

With the definitions of these vector spaces, we are ready to present our assumptions for the derivation of generalized solution of the forward problem.

\begin{assm}\label{imp_assump}
For $T>0$, we assume that 

\indent (A1) attack $D_a(.) \delta(t) \in L^2([0,T]; C^{2}(\Omega)\cap H^{-1}(\Omega))$,  and  $D_a(1)\neq 0$ \\
\indent  (A2) initial condition $\phi(x)\in C^{2}(\Omega)\cap H^{1}(\Omega)$ and $\phi'(0)=\phi'(1)=0$.

Assumption (A1) enables the attack to be discontinuous in time as long as it is Lebesgue measurable.
\end{assm}

Finally we are ready to define the idea of generalized solution of the PDE \eqref{pde}-\eqref{pde-output} in the following way.
\begin{deff}[Generalized Solution of \eqref{pde}-\eqref{pde-output}]
A function $u(.,t):[0,T] \to H^1_0(\Omega)$ where $\Omega = (0,1)$  is a weak solution of \eqref{pde}-\eqref{pde-output} if 
\begin{enumerate}[label=\roman*)]
    \item $u \in L^2([0,T];H^{1}_0(\Omega))$ and $u_t \in L^2([0,T];H^{-1}(\Omega))$;
    \item for any ``test function" $\Psi \in H^{1}_0(\Omega)) $ 
    \begin{align}\label{weak1}
        \langle u_t,\Psi\rangle + b(u,\Psi;t) =  \delta \langle D_a , \Psi\rangle, 
    \end{align}
   for  $t$  point-wise almost everywhere (a.e.) in $[0, T]$ and where $b(u,\Psi;t)$ is the bilinear form given by
    \begin{align}
        b(u,\Psi;t) = \int_0^1 u_x \Psi_x \, dx -\alpha\int_0^1  u \Psi\, dx.
    \end{align}
    \item $u(.,0)=\phi(.)$. 
\end{enumerate}
\end{deff}
The time derivative $u_t$ in \eqref{weak1} is defined as distributional time derivative: 
$$ \int_0^T u_t(.,t) v(t) \, dt = - \int_0^T u(.,t) v'(t) \, dt $$ for any $v: [0,T] \to \mathbb{R}$ with $v \in C_0^\infty(0, T)$. Moreover, a solution $u$ that satisfies \eqref{weak1}, automatically satisfies the condition $u_x(0,t)=u_x(1,t)=0$. This can be proved easily by applying Green's Theorem to the first term of the bilinear form and choosing $\Psi = \frac{\partial u}{\partial x}|_{\partial \Omega}$. Thus, no additional condition other than $u(.,0)=\phi(.)$ is necessary in order for $u$ to be a generalized solution of \eqref{pde}.
\begin{remm}
The integral equation \eqref{weak1} can be obtained by (i) multiplying \eqref{pde} with any  $\Psi \in H^{1}_0(\Omega))$, (ii) integrating between $(0,1)$ and (iii) applying divergence theorem. Thus the generalized solution represents the solution of the  integral version of the original PDE \eqref{pde}. The advantage of generalized solution compared to classical solution is that they allow solutions of much less restrictive class like $L^2([0,T];H^{-1}(\Omega))$ such as discontinuous attacks. 
\end{remm}

In order to construct the generalized solution of the forward problem, we follow the subsequent steps.
\noindent\\
\textit{Step 1:} We construct a sequence of approximate solutions to the desired PDE by constructing regular solutions for an ``approximate" version of the original PDE.\\
\textit{Step 2: }We show these solutions are uniformly bounded.\\
\textit{Step 3:} Using Banach-Alaoglu theorem, we can then claim the existence of a subsequence from this sequence of approximate solutions, such that they converge ``weakly" to a limit.\\
\textit{Step 4:} Finally, we argue that the limiting solution solves the PDE in the generalized sense.

\noindent
\textbf{Step 1: Construction of the Galerkin approximate solution:} 

Let $\chi^N$ be a $N-$dimensional subspace of $H^1_0(0,1)$ and is given by the span$\{X_1,X_2, \hdots, X_N\}$, where $\{X_1,X_2, \hdots\}$  are the orthonormal basis vectors of $L^2(0,1)$ as well as for $H^1_0(0,1)$. 
Next, let us define an orthogonal projection onto $\chi^N$, $\mathcal{P}_N: L^2(\Omega)\to \chi^N \subset L^2(\Omega)$  as
\begin{align}
    \mathcal{P}_N\left(\sum\limits_{n\in \mathbb{N}}T_nX_n(x)\right)= \sum\limits_{n\in \mathbb{N}}^NT_nX_n(x)
\end{align}
\begin{deff}[Approximate solution]
An approximate solution of \eqref{weak1} is a function $u_N:[0,T] \to \chi^N$, if \\
(i) $u_N\in L^2([0,T]; \chi^N)$ and $u_{Nt} \in L^2([0,T]; \chi^N) $ \\
(ii) $\forall w\in \chi^N$, 
\begin{align}
    &( u_{Nt},w)_{L^2} + b(u_N,w;t) =  \delta \langle P_N D_a, w\rangle,\label{approx_u}\\
    &u_N(0)=P_N\phi \label{approx_IC}
\end{align}
\noindent where $P_N D_a, P_Ng$ imples the orthogonal projection of functions $D_a$ and $g$ onto $\chi^N$  and are defined as below:
\begin{align}
  &D_a(x)= \sum\limits_{n\in \mathbb{N}} d_n X_n(x),\,\,P_ND_a(x)= \sum\limits_{n=0}^N d_n X_n(x), \\
  &\phi(x)= \sum\limits_{n\in \mathbb{N}} \phi_n X_n(x),\,\, P_N\phi(x)= \sum\limits_{n=0}^N \phi_n X_n(x).
\end{align}
 \end{deff}

Next using Sobolev Embedding theorem \cite{evans_PDE}, we can assert that if $u_N\in H^1_0(0,T; \chi^N)$ then $u_N\in C(0.T;\chi^N)$. For test functions in $\chi^N$, satisfying condition (ii) implies that the solution satisfies the weak formulation \eqref{weak1}.

In order to construct the ``Galerkin approximation", we choose the $N$  orthonormal basis vectors of $\chi^N$ to be the eigenfunctions corresponding to the eigenvalues $\lambda_n$ of the following Sturm-Liouville problem 
\begin{align}
    &X^{''}_n(x)+(\alpha -\lambda_n)X_n(x)=0, X'_n(0)=0, X'_n(1)=0,\label{Xeqn}
\end{align}
where
$
    \lambda_n=\alpha-(n\pi)^2,\quad X_n(x)=\cos(n\pi x).
$
For $u_N$ to be an approximate solution, $u_N$ must be in $L^2([0,T]; \chi^N)$ and thus can be expressed as the following sum of its basis vectors:
\begin{align}
    u_N(x,t)=\sum\limits_{n=0}^{N}T^N_n(t)X_n(x),\label{formal}
\end{align}
where $T^N_n(t):[0,T]\to \mathbb{R}$ are scalar coefficient functions that are absolute continuous over $[0,T]$. Moreover, \eqref{formal} is a solution of \eqref{approx_u} iff $T^N_n, (T^N_{n})_t\in L^2([0,T]), \forall n\in\{1, \hdots,N\}$. Plugging in \eqref{formal} into \eqref{approx_u}-\eqref{approx_IC} and using linearity along with \eqref{Xeqn}, we can assert that 
\begin{align}
    (T^N_n)'-\lambda_n T^N_n(t) = d_n \delta(t),\ T^N_n(0)=\phi_n. \label{T-sys}
\end{align}
where $d_n = \langle  D_a, X_n\rangle$ and $\phi_n=( \phi, X_n)_{L^2}$ as defined previously. To note, $\phi_n, d_n$,  represents Fourier coefficient of the even extension of $\phi$ and $D_a$ respectively.

Finally using the solution of \eqref{T-sys} into \eqref{formal}, we finally arrive at the approximate Galerkin solution
\begin{align}\label{approx_soln}
    u_N(x,t)=\sum\limits_{n=0}^{N}\phi_ne^{\lambda_n t}\cos(n\pi x) \nonumber + \sum\limits_{n=0}^{N}\left( \int\limits_0^td_n\delta(\tau)e^{\lambda_n(t-\tau)}d\tau\right)\cos(n\pi x).
\end{align}
\noindent
\textbf{Step 2: Uniform Bound on approximate solution:}
In this step, we invoke the following lemma demonstrating the uniform bound of the approximate solution $u_N$ for each $N \in \mathbb{N}$. The proof of this lemma is in the appendix.

\begin{Lemm}[Uniform Bound on approximate solution]\label{lemma_bound}
For every approximate Galerkin approximation $u_N$ there exists a constant $\gamma$ which depends on $T, \Omega$ and $\alpha$ such that it uniformly bounds the solution and its derivative in the following sense \cite{hunter}:
\begin{align}
    \|u_N\|_{L^\infty(0,T;L^2)}+\|u_N\|_{L^\infty(0,T;H_0^1)}+\|(u_N)_t\|_{L^2([0,T];H^{-1})} 
    \leqslant \gamma \left(\|D_a\delta\|_{L^2([0,T];H^{-1})} +\|\phi\|_{L^2(\Omega)}\right)
\end{align}
\end{Lemm}
\noindent
\textbf{Step 3: Convergence of approximate solution:}
Using the next Lemma, we establish the existence of a limiting weak solution derived from the sequence of approximate Galerkin solutions.

\begin{Lemm}[ Convergence of approximate solution]\label{approx_conv}
The sequence of approximate solutions $u_N$ converges weakly to a unique function  $u \in L^2([0,T];H_0^1) $ and $u_t\in L^2([0,T];H^{-1})$ as $N\to \infty$ and is given by:
\begin{align}
 u(x,t)=&\sum\limits_{n=0}^{\infty}\phi_ne^{\lambda_n t}\cos(n\pi x) \quad + \sum\limits_{n=0}^{\infty}\left( \int\limits_0^td_n\delta(\tau)e^{\lambda_n(t-\tau)}d\tau\right)\cos(n\pi x).\label{forward_soln}
\end{align}
\end{Lemm}
\begin{proof}
Since the approximate solutions $u_N$ are bounded in $L^2([0,T];H_0^1)$ (from Lemma \ref{lemma_bound}), we know from the Banach-Alaoglu theorem that there exists a subsequence of $u_N$ that weakly converges to a function in $L^2([0,T];H_0^1) $. Here we construct the subsequence to be the sequence $u_N$ itself and show that it weakly converges to function $u \in L^2([0,T];H_0^1)$ \eqref{forward_soln}  i.e. $u_N\rightharpoonup u$ in $L^2([0,T];H_0^1)$.

To prove weak convergence in $L^2([0,T];H_0^1)$, we need to show that (i) $L^2([0,T];H_0^1)$-norm of $(u_N-u)$ goes to zero as $N\to \infty$ and (ii) $u$ in \eqref{forward_soln} belongs to $L^2([0,T];H_0^1)$ .

Proving (i): We need to prove the following norm goes to zero as $N\to \infty$
$$
    \int\limits_0^T\!\!\|u_N-u\|^2_{H^1_0}dt = \int\limits_0^T\!\!\|u_N-u\|^2_{L^2}dt +\int\limits_0^T\!\!\|u_{Nx}-u_x\|^2_{L^2}dt
$$
Now expanding the second term of the above expression yields
$$
    \int\limits_0^T\!\!\|u_N-u\|^2_{L^2}dt \leqslant \sum\limits_{n=N}^{\infty}\left[ \int_0^T\phi^2_ne^{2\lambda_n t}dt\int_0^1\cos^2(n\pi x) dx\right]
     + \sum\limits_{n=N}^{\infty}\left[\int\limits_0^T\!\! \int\limits_0^t\!\!d^2_n\delta^2(\tau)e^{2\lambda_n(t-\tau)}d\tau\, dt \int_0^1\!\!\cos^2(n\pi x)dx\right].
$$

$\phi_n$ and $d_n$ being Fourier coefficients of $C^2$ functions on bounded interval are bounded and thus we can conclude that the right hand side of the inequality is a convergent sum and as $N\to \infty$, the norm goes to zero. We can similarly show that $\int_0^T\!\!\|u_{Nx}-u_x\|^2_{L^2}dt$ goes to zero as $N\to\infty$ which in turn proves part(i). 

Proving(ii):
To prove $u\in L^2([0,T];H_0^1)$, we need to show $\int_0^T\!\!\|u\|^2_{H^1_0}dt <\infty$. This again can be proved using similar convergence arguments, as used in part (i).
\end{proof}
\noindent
\textbf{Step 4: Validation of limiting solution: }

To prove that the limiting function is indeed the generalized solution, we present an important lemma regarding weak convergence.

\begin{Lemm}\label{strong_weak_conv}
If $u_N$ converges weakly to u ($u_N\rightharpoonup u$) in $L^2([0,T];H_0^1)$ and $\Psi_N$ converges strongly to $\Psi$ ($\Psi_N\to \Psi$) in $L^2([0,T];H^{-1})$, then $\langle u_N,\Psi_N\rangle \to \langle u,\Psi\rangle.$
\end{Lemm}
\begin{proof}
For any $\epsilon >0$, there exists a $N_1\in \mathbb{N}$ such that as $N>N_1$ we have
   $\int_0^T\!\left| \langle u_N, \Psi_N\rangle - \langle u, \Psi\rangle\right|\,dt 
     \leqslant \int_0^T\left| \langle u_N, \Psi_N\rangle - \langle u, \Psi_N\rangle \right|dt     +\int_0^T\left| \langle u, \Psi_N-\Psi\rangle \right|dt\! < \epsilon$. 

\end{proof}
Finally, we are ready to demonstrate that the function derived from the weak limit of the approximate Galerkin solutions is in fact the unique solutions of the generalized PDE problem formulated in \eqref{weak1}. 
\begin{thmm}[Generalized solution of forward problem]
The limiting function  $u$ in \eqref{forward_soln} is the unique generalized solution to \eqref{weak1}.
\end{thmm}
\begin{proof}
Using Lemma \ref{approx_conv}, we know that $u_N\rightharpoonup u$ in $L^2([0,T];H_0^1)$ and $u_{Nt}\rightharpoonup u_t$ in $L^2([0,T];H^{-1}).$ Next we choose a sequence of function $\Psi_N$ such that $\Psi_N\in \chi^N$ and substitute $\Psi_N$ in place of $w$ in \eqref{approx_u}. It is evident from the definition of $\chi^N$ that it is dense in $H_0^1(\Omega)$. This implies there exists a function $\Psi\in H_0^1(\Omega)$ such that $\Psi_N \to \Psi$ in $L^2([0,T];H^{-1})$. It is also trivial to prove that $\lim_{N\to \infty} P_N D_a = D_a$. 
Moreover, for any $h \in C_0^\infty(0,T)$ we can infer using Lemma \ref{strong_weak_conv} the following:
\begin{align}
    &\int_0^T ( u_{Nt},h \Psi_N)_{L^2}\,dt  \to \int_0^T \langle u_{t},h \Psi\rangle\,dt   .\\
    &\int_0^T b(u_N,h \Psi_N;t) \,dt  \to  \int_0^T b(u,h \Psi;t)\\
    &  \int_0^T \delta \langle P_N D_a, h \Psi_N\rangle \,dt   \to \int_0^T  \delta \langle D_a, h \Psi \rangle\,dt,
\end{align}
as map $t \mapsto h \Psi_N$ is in  $L^2([0,T];H_0^1)$. Hence, if $u_N$ satisfies \eqref{approx_u} as $N\to \infty $, $u$ satisfies the following equation:
\begin{align}\label{almost_FS}
    \int_0^T h  \left(\langle u_{t},\Psi\rangle+b(u, \Psi;t)\right)\,dt = \int_0^T h \delta \langle D_a,  \Psi\rangle. 
\end{align}
Considering \eqref{almost_FS} is true for any $h \in C_0^\infty(0,T)$, we can conclude that $u$ satisfies the generalized PDE \eqref{weak1}.

\end{proof}

\subsection{Existence and Uniqueness of Stealthy Attacks}
In this subsection, we use the derived generalized solution of the forward problem to construct a solution of the inverse problem in the form of $\delta(t)$. Subsequently, we prove the existence and uniqueness of this solution.

\begin{thmm}[Existence and Uniqueness of Stealthy Attacks]\label{theorem1}
Consider the PDE system \eqref{pde}-\eqref{pde-output} that satisfies Assumption \ref{imp_assump}, with $q(t)=0$ and the stealthiness condition \eqref{stealthiness}.
Then there exists an unique solution $\delta(t)\in L^2[0,T]$ with $0<T<\infty$ for which $y(t)=u(1,t)=0$.
\end{thmm}
\begin{proof}

In this part, we solve the following inverse problem: with the knowledge of $y(t)=u(1,t)=0$ and the system dynamics, solve for $\delta(t)$. Substituting \eqref{pde-output} in \eqref{forward_soln}, we get the following Volterra integral equation of the first kind with respect to $\delta(t)$:
\begin{align}\label{Volterra}
    u(1,t)=0=-a(t)+\int\limits_0^t \delta(\tau)b(t,\tau)\,d\tau, 
\end{align}
where 
\begin{align}\label{a}
    a(t):=-\sum\limits_{n=0}^{\infty}(-1)^{n}\phi_ne^{\lambda_n t},
    b(t,\tau):=\sum\limits_{n=0}^{\infty}(-1)^{n} d_n e^{\lambda_n(t-\tau)}.
\end{align}
As $d_n$ and $\phi_n$ are Fourier coefficient, we can again argue that the series given in $a(t)$ and $b(t,\tau)$ converge and are continuously differentiable functions in $t$ under assumptions (A1) and (A2). 

Next, we prove that $\lim_{t\to 0}a_n(t)$ exists. This is needed to show that \eqref{Volterra} is valid as $t\to 0$. Now, the existence of $\lim_{t\to 0}a_n(t)$ is equivalent to the uniform convergence of the sum of limits for each term. Performing term-by-term limit of \eqref{a} yields
\begin{align}
    \lim_{t\to 0}\sum\limits_{n=0}^{\infty}(-1)^{n}\phi_ne^{\lambda_n t}&=\sum\limits_{n=0}^{\infty}(-1)^{n}\phi_n\left(\lim_{t\to 0}e^{\lambda_n t}\right)=\sum\limits_{n=0}^{\infty}(-1)^{n}\phi_n. \label{a0}
\end{align}
In order to show uniform convergence of \eqref{a0}, we need to find a convergent majorizing series. Using integration by parts and condition $\phi'(0)=\phi'(1)=0$ in (A2), we can show that $|\phi_n|<\frac{2C_\phi}{n^2\pi^2}$. This implies the majorizing series $\sum\limits_{n=0}^{\infty}\frac{2C_\phi}{n^2\pi^2}$ is absolutely convergent. Hence $\sum\limits_{n=0}^{\infty}(-1)^{n}\phi_n$ converges uniformly. This implies that the limit of \eqref{a} as $t\to0$ exists and can be obtained by summing the limits of the terms. 
Finally, we evaluate the limit of $a(t)$ as $t\to 0$ we get
\begin{align}
    a(0)=-\sum\limits_{n=0}^{\infty}(-1)^{n}\phi_n= -\phi(1)=-u(1,0)=0.
\end{align}
Hence, it is evident that  \eqref{Volterra} is valid as $t\to 0$. 

Next, using similar arguments we can evaluate the limit of $t\to \tau$ for $b(t,\tau)$ to obtain
\begin{align}
    b(t,t)= \sum\limits_{n=0}^{\infty}(-1)^{n}d_n =D_a(1)\neq 0, \quad t\geqslant 0 \label{b_not0}
\end{align}
using the condition (A2). This is needed subsequently in the proof.

Furthermore, under the conditions (A1) and (A2) and the properties of $a(t)$ and $b(t,\tau)$ derived above, we can differentiate \eqref{Volterra} to obtain a Volterra integral equation of the second kind as follows:
\begin{align}
    \delta(t)b(t,t)+\int\limits_0^t\delta(\tau)b_t(t,\tau)\, d\tau=a'(t),\label{volterra2}
\end{align}
which can be equivalently written as
\begin{align}
    \delta(t)-\int\limits_0^tK(t,\tau) \delta(\tau)\, d\tau=m(t),\label{volterra22}
\end{align}
where
\begin{align}
    &m(t)=\frac{a_t(t)}{b(t,t)}\quad \text{ and }\quad K(t,\tau)= -\frac{b_t(t,\tau)}{b(t,t)}  \label{m_k}\\
    &b_t(t,\tau)= \sum\limits_{n=0}^\infty (-1)^n\lambda_n d_ne^{\lambda_n(t-\tau)}. \label{b_t}
\end{align}
Let us define an integral operator $\mathcal{K}$ with the kernel $K(t,\tau)$ as follows
\begin{align}
     \mathcal{K}m(t)=\int\limits_0^t K(t,\tau)m(\tau)\,d \tau.
\end{align}
Since $b(t,t)=D_a(1)$ is a non-zero constant function given in \eqref{b_not0}, using standard iterative methods \cite{hildebrand} we obtain a \textit{unique} function $\delta(t)$ that is the solution of \eqref{volterra2}. The solution is formally given by the  following  Liouville-Neumann series:
\begin{align}
    \delta(t)= m(t)+\sum\limits_{n=1}^\infty \mathcal{K}^nm(t). \label{delta}
\end{align} 

Now, to assert the $L^2$-convergence of the formal series given in \eqref{delta}, we need to show that (i) the kernel $K(t,\tau)\in L^2(I,I)$ where $I=[0,T], T<\infty$, (ii) the operator norm of $\mathcal{K}$ is finite over a finite interval I and (iii) sum in \eqref{delta} converges in $L^2(I)$.

To prove the kernel is square-measurable we need to show $\int_0^T\int_0^t |K(t,\tau)|^2 d\tau dt < \infty$. Now $\|K(t,.)\|_{L^2}^2=\int_0^t |k(t,\tau)|^2 d\tau = \frac{1}{|b(t,t)|^2}\int_0^t |b_t(t,\tau)|^2d \tau$. Using Cauchy product for $|b_t(t,\tau)|^2$ and observing $|b_t(t,\tau)|^2\leqslant |b_t(t,t)|^2$ (since $\lambda_n<0$), we can write $\|K(t,.)\|_{L^2}^2 <Ct$, where $C = |b_t(t,t)|^2/|b(t,t)|^2$. This implies $\|K\|_{L^2}^2 = CT^2/2<\infty$. 

To prove that the operator $\mathcal{K}:L^2(I)\to L^2(I)$ have a finite operator norm, we chose $m(t)\in L^2(I)$, then $\|\mathcal{K}\|^2=\sup_{\|m\|} \frac{\|\mathcal{K}m\|^2_{L^2}}{\|m\|^2_{L^2}}< C\int_0^T \frac{t_1^2}{2} dt_1=C\frac{T^3}{3!}<\infty$. This means that the operator maps from the space of measurable functions to itself.

Similarly, it is straightforward to verify that  $\|\mathcal{K}^n m\|^2_{L^2}< C\|m\|^2_{L^2}\frac{T^{n+2}}{(n+2)!}$. Therefore, we can write $\|\delta\|_{L^2}\leqslant \sum_{n=0}^\infty\|\mathcal{K}^n m\|_L^2 \leqslant \|m\|_{L^2} T  \sum_{n=0}^\infty \sqrt{\frac{(CT)^n}{(n+2)!}}.$ Using ratio test it can be easily verified that this series is convergent. 

From these three arguments we can infer that the series of measurable functions mapped by the operator converges to a measurable function $\delta(t)$. Hence, the existence of a solution of the inverse problem is proven by construction of a series solution \eqref{delta} to the inverse problem. The iterative method of this construction also guarantees uniqueness of such construction. This concludes the proof of \textit{existence} of a  \textit{unique} solution $\delta(t)$ to the inverse problem \eqref{pde}-\eqref{pde-output}. 
\end{proof}

Once we find that there exists a unique solution of the inverse problem that provides $\delta(t)$, our next objective is to explore the stability of such solution in $L^2$. We know that a solution is \textit{locally well-posed} in the sense of Hadamard \cite{Lee} if the following conditions are satisfied for some time $T>0$ : (i) a solution exists, (ii) it is unique, and (iii) it ``continuously" depends on the initial data. 

Now let us define a linear operator $\mathcal{T}\delta(t):=\int\limits_0^t \delta(\tau)b(t,\tau)\,d\tau $). Next referring to \eqref{Volterra}, we observe that the linear problem  $\mathcal{T}\delta(t) = a(t)$ is ill-posed as $\mathcal{T}$ is not continuous but $L^2$. Hence, even though the solution $\delta(t)$ exists and is unique, the solution remains ill-posed. This implies that the computation of such solution will not provide a correct answer directly regardless of the uniqueness of the solution. Accordingly, our next objective is to find a stable solution using Tikhonov Regularisation so that the solution  $\delta(t)$ obtained in \eqref{delta}, ``continuously" depends on the initial data in the sense of $L^2$-regularity. That is, for arbitrarily small changes in initial data $\phi(x)$, the change in $\delta$ will also be small in $L^2$ sense. To this end, we first introduce the regularisation method.

\begin{Lemm}[Tikhonov Regularisation \cite{Tikhonov_Hansen}]\label{tikhonov_lemma}
The stable solution to an ill-posed linear problem $\mathcal{T}m(t)=\delta(t)$, where $T:L^2(I)\to L^2(I)$ and $m(t),\delta(t)\in L^2(I)$ is given by
\begin{align}
    \delta_\gamma = \arg\!\!\!\!\!\min_{w\in L^2(I)}\|Tw -a\|^2_{L^2}+\gamma\|w\|^2_{L^2},
\end{align}
where $\gamma$ is a positive regularization parameter.
\end{Lemm}

\begin{thmm}[Continuous dependence of stealthy attack upon initial data] Let $\phi$ and $\tilde{\phi}$ be the nominal and perturbed initial conditions that satisfy the assumptions (A1) and (A2) of Theorem \ref{theorem1}, then the change in attack $\delta$ to $\tilde{\delta}$ depends continuously upon the change in initial data. 
\end{thmm}

\begin{proof}
Let us first denote $a$ and $\tilde{a}$ from \eqref{a} corresponding to the nominal and perturbed initial conditions.  Then the two linear problems becomes: $\mathcal{T}\delta = a$  and $\mathcal{T}\tilde{\delta} = \tilde{a}$. Using the Tikhonov regularization from Lemma \ref{tikhonov_lemma}, we can claim the regularised solutions to the two ill-posed linear problem are given by $\delta_\gamma$ to $\tilde{\delta}_\gamma$ respectively. Now, we can write
\begin{align}
    \|\delta - \tilde{\delta}\|_{L^2} \leqslant  \|\delta-\delta_\gamma\|_{L^2}+ \|\delta_\gamma - \tilde{\delta}_\gamma\|_{L^2}+ \|\tilde{\delta}_\gamma - \tilde{\delta}\|_{L^2}.
\end{align}
The first and the third terms on the right hand side of this inequality are due to approximation errors occurring from the approximation in the regularization method. Thus these terms can be bounded by a constant. Then it is clear that $\|\delta - \tilde{\delta}\|_{L^2}$ solely depends on the second term.

On the other hand, from \cite{Tikhonov_Hansen} we can claim that 
\begin{align}
    \|\delta_\gamma - \tilde{\delta}_\gamma\|_{L^2} \leqslant \gamma^{-1}\| a- \tilde{a}\|_{L^2}.
\end{align}
and it can be easily derived that 
\begin{align}
    \| a- \tilde{a}\|_{L^2} \leqslant M\| \phi- \tilde{\phi}\|_{C(\Omega)}.
\end{align}
This completes the proof.
\end{proof}

\section{DESIGN OF ATTACK DETECTION ALGORITHM}
In this section, we discuss the detailed design of attack detection algorithm for detectable attacks. We choose the following output error injection based PDE observer as the attack detection algorithm $\mathcal{A}(.)$.
\begin{align}
    &\hat{u}_t(x,t) = \hat{u}_{xx}(x,t) + \alpha \hat{u}(x,t) + D(x)q(t) + L(x)\tilde{y}(t),\\ 
    &\hat{u}_x(0,t) = 0, \hat{u}_x(1,t) = L_1\tilde{y}(t), \hat{u}(x,0) = \hat{\phi}(x), \\
    &\hat{y}(t)=\hat{u}(1,t),\quad r(t) = \tilde{y}(t) = \tilde{u}(1,t),\label{res-def} 
\end{align}
where $\tilde{y}(t) = y(t)-\hat{y}(t)$ is the output error of the observer, and $L(x)$ and $L_1$ are the observer gains and tuning parameters of the algorithm. 

\begin{remm}
We resort to a \textit{late lumping} approach where the observer design is done in an infinite dimensional space, as opposed to designing it on a finite dimensional approximation of the original PDE. Such \textit{late lumping} approach helps the proposed design to eliminate the following issues associated with finite dimensional approximation: (i) spillover effect \cite{Balas1979}, and (ii) loss of accuracy and physical significance of the original PDE model \cite{Ferdowsi2014}.
\end{remm}

In practical settings, the attack detection algorithm and its residual signal should have the certain desirable properties. First, stability is the primary requirement for the residual signal dynamics. Next, the residual signal should be robust with respect to model uncertainties. Finally, the residual signal should be highly sensitive to the attacks. 

Besides robustness to uncertainties, another relatively less obvious challenge in attack detection lies in distinguishing cyber attacks from physical faults. Although cyber attacks and physical faults both can negatively impact the system, the subsequent management could be very different for these two types of anomalies. Hence, it is essential to distinguish their occurrences, if possible. However, under certain scenarios it might be impossible to distinguish them, for example, if the adversary deliberately designs the attack to mimic a physical fault scenario. Apart from these scenarios, we resort to the following approach to distinguish the faults from attacks \cite{li2019anomaly}. We consider the physical faults as undesirable disturbances and subsequently desensitize their effect on the residual signal. In effect, we design our residual signal generator to be robust towards both uncertainties and physical faults. In order to do the same, knowledge about probable physical faults in system would be useful which can be acquired by performing a thorough apriori physical Failure Modes and Effect Analysis (FMEA) study \cite{stamatis2003failure}. Accordingly, the proposed algorithm will only raise a flag when there is a cyber-attack. In the next subsections, we formulate the observer design problem that considers these aforementioned properties.

\begin{remm}
The proposed framework allows for a provision to distinguish between cyber-attacks and physical faults. One possible solution can be running this algorithm in parallel with another fault detection algorithm. In case of a cyber-attack, the proposed algorithm will raise a flag whereas the fault detection algorithm will not. By comparing the outputs of these two algorithms, one can distinguish between cyber-attacks and faults.
\end{remm}

\begin{remm}
The attacks considered in this section are \textit{detectable attacks}. That is, they exclude the attack set discussed in the previous section. The analysis in the previous section discusses existence of undetectable attacks which cannot be detected by the algorithm $\mathcal{A}(.)$. Indeed, it is possible to generate such perfectly stealthy attacks if the adversary possesses detailed knowledge of the system, as discussed in \cite{pasqualetti2013attack}. On the other hand, it is still difficult to detect \textit{detectable attacks} that are generated based on partial knowledge of the system due to the presence of uncertainties. The goal of $\mathcal{A}(.)$ discussed in this section is to minimizes the effect of such uncertainties and maximizes the detectability of such \textit{detectable attacks}.
\end{remm}

\subsection{Residual Signal Generation and Design Requirements}
First, we re-write the PDE dynamics \eqref{pde}-\eqref{pde-output} as:
\begin{align}
    &u_t(x,t) = u_{xx}(x,t) + \alpha u(x,t) + D(x)q(t) + D_a(x)\delta(t) + \eta(x,t), \label{pde-new}\\
    &u_x(0,t) = 0, \quad u_x(1,t) = 0, \quad u(x,0) = \phi(x),  \label{pde-bc-new}\\
    &y(t)=u(1,t),  \label{pde-output-new}
\end{align}
where $\eta(x,t)$ represents the effect of distributed modeling uncertainty and physical faults. Subtracting the observer dynamics \eqref{res-def} from the system dynamics \eqref{pde-new}-\eqref{pde-bc-new}, we write the residual signal dynamics as
\begin{align}
    &\tilde{u}_t(x,t) = \tilde{u}_{xx}(x,t) + \alpha \tilde{u}(x,t)+ D_a(x)\delta(t) + \eta(x,t) - L(x)\tilde{u}(1,t), \label{pde-err}\\
    &\tilde{u}_x(0,t) = 0, \tilde{u}_x(1,t) = -L_1\tilde{u}(1,t),  \label{pde-bc-err}\\
    & r(t) = \tilde{u}(1,t),\label{residual-2}
\end{align}
where $\tilde{u}(x,t)=u(x,t)-\hat{u}(x,t)$ is the distributed error with $\tilde{u}(x,0) = \tilde{u}_0(x) = \phi(x)-\hat{\phi}(x)$ being the initial error. Next, we state the following design requirements.

\begin{DR}\normalfont(Exponential Stability)
In case of $\eta(x,t)=0, \delta(t)=0$, the residual should exponentially approach zero. We mathematically express this requirement as:
\begin{equation}
\left|r(t)\right|^2 \leqslant K_1 e^{-K_2t} \left|r(0)\right|^2, \label{dr1}
\end{equation}
where $K_1 \in \mathbb{R}$ and $K_2 \in \mathbb{R}^+$ are some constants.
\end{DR}

\begin{DR}\normalfont(Attack/Uncertainty-to-Residual Stability)
In case of $\eta(x,t) \neq 0, \delta(t) \neq 0$, the residual signal should be bounded. We mathematically express this requirement as:
\begin{align}
\left|r(t)\right|^2 \leqslant K_3 e^{-K_4t} \left|r(0)\right|^2 + K_5(t)\left|\delta(t)\right|^2 +  K_6\left\|\eta(.,t)\right\|_{\mathcal{H}}^2, \label{dr2}
\end{align}
where $K_3,K_4,K_5,K_6 \in \mathbb{R}^+$.
\end{DR}

\begin{DR}\normalfont(Robustness)
The residual signal should be robust with respect to the uncertainty and physical faults. We mathematically express this requirement as:
\begin{align}
\int_0^{\infty} \left|r(\tau)\right|^2 d\tau \leqslant \beta_1^2 \int_0^{\infty} \left\|\eta(.,\tau)\right\|_{\mathcal{H}}^2 d\tau + \epsilon,\quad \eta(x,t) \neq 0, \delta(t) = 0,  \label{dr3}
\end{align}
where $\beta_1 \in \mathbb{R}^+$ is an user-defined constant representing desired robustness level, and $\epsilon \in \mathbb{R}^+$ is a constant that depends on initial observer errors.
\end{DR}

\begin{DR}\normalfont(Attack Sensitivity)
The residual signal should be sensitive to the attack. We mathematically express this requirement as:
\begin{align}
\int_0^{\infty} \left|r(\tau)\right|^2 d\tau \geqslant  \beta_2^2 \int_0^{\infty} \left\|D_a(.,\tau)\right\|_{\mathcal{H}}^2\left|\delta(\tau)\right|^2 d\tau - \epsilon,\eta(x,t) = 0, \quad \delta(t) \neq 0,  \label{dr4}
\end{align}
where $\beta_2 \in \mathbb{R}^+$ is an user-defined constant representing desired sensitivity level, and $\epsilon \in \mathbb{R}^+$ is a constant that depends on initial observer errors.
\end{DR}

\begin{remm}
The Design Requirements 3 and 4 are inspired by the $\mathcal{H}_{\infty}/\mathcal{H}_{-}$ criteria used in fault detection of ODE systems \cite{wang2007lmi}. Furthermore, the term $\epsilon$ in \eqref{dr3} and \eqref{dr4} can be interpreted following the notion of Input-to-State Practical stability \cite{sontag1996new}. 
\end{remm}

\subsection{Backstepping Transformation of Residual Dynamics}
In this section, we use a backstepping transformation to transform the residual dynamics \eqref{pde-err}-\eqref{pde-bc-err} to a target system dynamics \cite{krstic2008boundary}. Such target system dynamics will be useful for analyzing and designing the residual signal properties. The following transformation $\tilde{u}(x,t) \mapsto \psi(x,t)$
\begin{equation}
    \tilde{u}(x,t) = \psi(x,t) - \int_x^1 P(x,y)\psi(y,t)dy,\label{bst}
\end{equation}
transforms \eqref{pde-err}-\eqref{pde-bc-err} to the following target system dynamics:
\begin{align}
    &\psi_t(x,t) = \psi_{xx}(x,t) -c \psi(x,t)+ D(x)\delta(t)+ \theta(x,t), \label{pde-err-new}\\
    &\psi_x(0,t) = 0, \psi_x(1,t) = -c\psi(1,t),  \label{pde-bc-err-new}
\end{align}
where $c>0$ is a design parameter to be determined, and $D(x)$ and $\theta(x,t)$ are the transformed versions of $D_a(x)$ and $\eta(x,t)$, respectively, given by
\begin{align}
    &\eta(x,t) = \theta(x,t) - \int_x^1 P(x,y)\theta(y,t)dy,\label{eta-tr}\\
    &D_a(x) = D(x) - \int_x^1 P(x,y)D(y)dy.\label{Da-tr}
\end{align}
Starting from the transformation \eqref{bst}, and subsequently comparing the original residual dynamics \eqref{pde-err}-\eqref{pde-bc-err} and the target system \eqref{pde-err-new}-\eqref{pde-bc-err-new}, it turns out that the kernel $P(x,y)$ must satisfy the following PDE \cite{krstic2008boundary}:
\begin{align}
    &P_{yy}(x,y)-P_{xx}(x,y)=(c+\alpha)P(x,y),\label{Pt1}\\
    &P(x,x) = -\frac{(c+\alpha)}{2}x,\quad P_x(0,y)=0.\label{Pt2}
\end{align}
with the following closed-form solution
\begin{equation}
    P(x,y)=-(c+\alpha)y\frac{I_1(\sqrt{(c+\alpha)(y^2-x^2)})}{\sqrt{(c+\alpha)(y^2-x^2)}},\label{PT-sol}
\end{equation}
where $I_1(.)$ is the modified Bessel function of first kind. Furthermore, the observer gains can be computed as:
\begin{equation}
    L(x) = -cP(x,1)-P_y(x,1), \quad L_1 = c-P(1,1).\label{obs-gains}
\end{equation}

Next, we focus on the inverse transformation $Q(x,y) \mapsto P(x,y)$ which is given by
\begin{equation}
    \psi(x,t) = \tilde{u}(x,t) + \int_x^1 Q(x,y)\tilde{u}(y,t)dy.\label{bst-inv}
\end{equation}
The condition on the kernel $Q(x,y)$ is derived following a similar approach as the backstepping kernal $P(x,y)$. It turns out that $Q(x,y)$ must satisfy the following PDE:
\begin{align}
    &Q_{yy}(x,y)-Q_{xx}(x,y)=-(c+\alpha)Q(x,y),\label{Pt1-inv}\\
    &Q(x,x) = -\frac{(c+\alpha)}{2}x,\quad Q_x(0,y)=0.\label{Pt2-inv}
\end{align}
with the following closed-form solution
\begin{equation}
    Q(x,y)=-(c+\alpha)y\frac{J_1(\sqrt{(c+\alpha)(y^2-x^2)})}{\sqrt{(c+\alpha)(y^2-x^2)}},\label{PT-inv-sol}
\end{equation}
where $J_1(.)$ is the Bessel function of first kind. Consequently, the inverse relationships of \eqref{eta-tr} and \eqref{Da-tr} can be written as:
\begin{align}
    &\theta(x,t) = \eta(x,t) + \int_x^1 Q(x,y)\eta(y,t)dy,\label{eta-tr-inv}\\
    &D(x) = D_a(x) + \int_x^1 Q(x,y)D_a(y)dy.\label{Da-tr-inv}
\end{align}

\begin{remm}\label{invertibility}
Due to the existence of an inverse transformation \eqref{bst-inv}, we can conclude that the stability properties of target system dynamics \eqref{pde-err-new}-\eqref{pde-bc-err-new} confirm the stability of the original system dynamics \eqref{pde-err}-\eqref{pde-bc-err}.
\end{remm}

\begin{Lemm}\label{invertibility2}
The upper bounds of $Q(x,y)$ and $Q_x(x,y)$ are given by:
\begin{align}
    &\left|Q(x,y)\right| \leqslant \overline{Q}= \sup_{0 \leqslant x \leqslant y \leqslant 1}\left|Q(x,y)\right| = \frac{c+\alpha}{2},\label{ub11}\\
    &\left|Q_x(x,y)\right| \leqslant \overline{Q}_x= \sup_{0 \leqslant x \leqslant y \leqslant 1}\left|Q_x(x,y)\right| = \frac{c+\alpha}{16}.\label{ub21}
\end{align}
\end{Lemm}
\begin{proof}
Staring from \eqref{PT-inv-sol} and using properties of Bessel functions \cite{krstic2008boundary}, the upper bounds can be found equating the first derivative of $Q(x,y)$ and $Q_x(x,y)$ to zero, and subsequently computing the supremum values of the functions.  
\end{proof}

\begin{Lemm}\label{D-bound}
The upper bounds of $D(x)$ and $\theta(x,t)$ in \eqref{pde-err-new} are given by:
\begin{align}
    & \left|D(1)\right| = \left|D_a(1)\right|,\left\|D(.)\right\| \leqslant \Big(1+\frac{c+\alpha}{2}\Big)\left\|D_a(.)\right\|,\\
    & \left\|D_x(.)\right\| \leqslant \Big(1+\frac{c+\alpha}{16}\Big)\left\|D_{ax}(.)\right\|,\\
    & \left|\theta(1,t)\right| \leqslant \left|\eta(1,t)\right|, \left\|\theta(.,t)\right\| \leqslant \Big(1+\frac{c+\alpha}{2}\Big)\left\|\eta(.,t)\right\|,\\
    & \left\|\theta_x(.,t)\right\| \leqslant \Big(1+\frac{c+\alpha}{16}\Big)\left\|\eta_x(.,t)\right\|,\label{D-bnd1}
\end{align}
\end{Lemm}
\begin{proof}
Staring from \eqref{eta-tr-inv}-\eqref{Da-tr-inv}, and using the results from Lemma \ref{invertibility2} and Cauchy-Schwarz inequality, the proof is straightforward.
\end{proof}

\subsection{Main Result of Attack Detection Algorithm}

In this section, we present the main result of the attack detection algorithm. First, we make the following assumptions.

\begin{assm}\label{as-init}
Consider the target system \eqref{pde-err-new}-\eqref{pde-bc-err-new}. The initial conditions are bounded as follows: 
\begin{align}\label{init-bnd}
    & \left\|\psi(x,0)\right\|\leqslant \overline{\psi}_0<\infty,
    \left\|\psi_x(x,0)\right\|\leqslant \overline{\psi}_{x0}<\infty , \\
    &  0<\underline{\psi}_{10}  \leqslant \left|\psi(1,0)\right|\leqslant \overline{\psi}_{10} <\infty.
\end{align}
\end{assm}

Next, the following theorem establishes the conditions on the design parameter $c$ in target system \eqref{pde-err-new}-\eqref{pde-bc-err-new} which will lead to the satisfaction of the Design Requirements 1-4.

\begin{thmm}
Consider the residual signal definition \eqref{res-def}, Design Requirements \eqref{dr1}-\eqref{dr4} and the target system dynamics \eqref{pde-err-new}-\eqref{pde-bc-err-new}. If there exists a $c > \lambda>0$ that satisfies the following conditions:
\begin{equation}
\mathbb{A} \leqslant0, \quad \mathbb{B} \geqslant 0, \label{cond2}
\end{equation}
where $\lambda$ being an arbitrary positive constant and $ \mathbb{A} = [\mathbb{A}_{ij}] \in \mathbb{R}^{6\times 6}$ with
\begin{align}
&\mathbb{A}_{11} = 1-c^2+\frac{c\alpha_1\lambda_1}{2},
\mathbb{A}_{22} = -c+(1+\frac{c+\alpha}{2})\frac{\alpha_2\lambda_2}{2},\nonumber\\
&\mathbb{A}_{33} = -c+(1+\frac{c+\alpha}{16})\frac{\alpha_3\lambda_3}{2},
\mathbb{A}_{44} = \frac{c\alpha_1}{2\lambda_1}-\beta_1^2,\nonumber\\
&\mathbb{A}_{55} = (1+\frac{c+\alpha}{2})\frac{\alpha_2}{2\lambda_2}-\beta_1^2,
\mathbb{A}_{66} = (1+\frac{c+\alpha}{16})\frac{\alpha_3}{2\lambda_3}-\beta_1^2,\nonumber\\
&\mathbb{A}_{14} = \mathbb{A}_{41} = \frac{c(1-\alpha_1)}{2},
\mathbb{A}_{25} = \mathbb{A}_{52} = \frac{(1+\frac{c+\alpha}{2})(1-\alpha_2)}{2},\nonumber\\
&\mathbb{A}_{36} = \mathbb{A}_{63} = \frac{(1+\frac{c+\alpha}{16})(1-\alpha_3)}{2},
\label{cond2a}
\end{align}
and rest of the elements $\mathbb{A}_{ij} = 0$; $ \mathbb{B} = [\mathbb{B}_{ij}] \in \mathbb{R}^{6\times 6}$ with
\begin{align}
&\mathbb{B}_{11} = 1+c^2-\frac{c\alpha_4\lambda_4}{2},
\mathbb{B}_{22} = +c-(1+\frac{c+\alpha}{2})\frac{\alpha_5\lambda_5}{2},\nonumber\\
&\mathbb{B}_{33} = +c-(1+\frac{c+\alpha}{16})\frac{\alpha_6\lambda_6}{2},
\mathbb{B}_{44} = -\frac{c\alpha_4}{2\lambda_4}-\beta_2^2,\nonumber\\
&\mathbb{B}_{55} = -(1+\frac{c+\alpha}{2})\frac{\alpha_5}{2\lambda_5}-\beta_2^2,
\mathbb{B}_{66} =- (1+\frac{c+\alpha}{16})\frac{\alpha_6}{2\lambda_6}-\beta_2^2,\nonumber\\
&\mathbb{B}_{14} = \mathbb{B}_{41} = -\frac{c(1-\alpha_4)}{2},
\mathbb{B}_{25} = \mathbb{B}_{52} =- \frac{(1+\frac{c+\alpha}{2})(1-\alpha_5)}{2},\nonumber\\
&\mathbb{B}_{36} = \mathbb{B}_{63} = -\frac{(1+\frac{c+\alpha}{16})(1-\alpha_6)}{2},
\label{cond3a}
\end{align}
and rest of the elements $\mathbb{B}_{ij} = 0$ with $\alpha_i \in \mathbb{R}$, $\lambda_i >0$, $i \in \{1,2,..,6\}$;\\
then the following statements (P1)-(P4) are true:\\\\
\textit{\textbf(P1)} In the absence of uncertainty and attack, i.e. $\eta(x,t)=0,\delta(t)=0$, the residual signal will exponentially converge to zero as follows:
\begin{equation}
    \left|r(t)\right|^2 \leqslant K_1 e^{-K_2t} \left|r(0)\right|^2,\label{part1} 
\end{equation}
where $K_1 = \frac{2}{c}(\frac{c}{2}+\frac{\overline{\psi}_0+\overline{\psi}_{x0}}{2\underline{\psi}_{10}})$, and $K_2 = 2c$.\\\\
\textit{\textbf(P2)} In the presence of uncertainty and attack, i.e. $\eta(x,t)\neq 0,\delta(t) \neq 0$, the residual signal is Input-to-State (ISS) stable with respect to the attack and uncertainty as follows:
\begin{align}
\left|r(t)\right|^2 \leqslant K_3 e^{-K_4t} \left|r(0)\right|^2 + K_5 \sup_{t \geqslant 0} \Big(\left\|D_a(.)\right\|_{\mathcal{H}}^2\left|\delta(t)\right|^2 + \left\|\eta(.,t)\right\|_{\mathcal{H}}^2\Big), \label{part2}
\end{align}
where $K_3 = \frac{2}{c}\Big(\frac{c}{2}+\frac{\overline{\psi}_0+\overline{\psi}_{x0}}{2\underline{\psi}_{10}}\Big), K_4 = 2(c-\lambda), K_5 = \frac{M_2}{c(c-\lambda)}$, and 
\begin{align}
   &M_2 =\max\Big\{\frac{c^2}{2\gamma},\frac{\big(1+\frac{c+\alpha}{2}\big)^2}{2\gamma},\frac{(1+\frac{c+\alpha}{16})^2}{2\gamma}\Big\}.\label{nn1} 
\end{align}
\\
\textit{\textbf(P3)} The residual signal satisfies the robustness requirement \eqref{dr3} as
\begin{equation}\label{part3}
\int_0^{\infty} \left|r(\tau)\right|^2 d\tau \leqslant \beta_1^2 \int_0^{\infty} \left\|\eta(.,\tau)\right\|_{\mathcal{H}}^2 d\tau + \epsilon,  
\end{equation}
with $\eta(x,t) \neq 0, \delta(t) = 0$ where $\epsilon = \frac{c}{2}\overline{\psi}_{10}^2 + \overline{\psi}_0^2 + \overline{\psi}_{x0}^2$.\\\\
\textit{\textbf(P4)} The residual signal satisfies the attack sensitivity requirement \eqref{dr4} as
\begin{equation}\label{part4}
\int_0^{\infty} \left|r(\tau)\right|^2 d\tau \geqslant  \beta_2^2 \int_0^{\infty} \left\|D_a(.,\tau)\right\|_{\mathcal{H}}^2\left|\delta(\tau)\right|^2 d\tau - \epsilon,  
\end{equation}
with $\eta(x,t) = 0, \delta(t) \neq 0$ where $\epsilon = \frac{c}{2}\overline{\psi}_{10}^2 + \overline{\psi}_0^2 + \overline{\psi}_{x0}^2$.
\end{thmm}
\begin{proof}
Considering the target system dynamics \eqref{pde-err-new}-\eqref{pde-bc-err-new}, we start with the following Lyapunov functional candidate:
\begin{equation}
    W(t) = \frac{c}{2}\left|\psi(1,t)\right|^2 + \frac{1}{2}\left\|\psi(.,t)\right\|^2 + \frac{1}{2}\left\|\psi_x(.,t)\right\|^2. \label{lyap1} 
\end{equation}
Differentiating $W$ with respect to $t$ along the trajectory of \eqref{pde-err-new}-\eqref{pde-bc-err-new}, and applying integration by parts and Cauchy-Schwarz inequality multiple times, we can write the upper bound of $\dot{W}$ as
\begin{align} \label{lyap8} \nonumber
    \dot{W} \leqslant-c^2\left|\psi(1)\right|^2 - c\left\|\psi_{x}\right\|^2 - c\left\|\psi\right\|^2 
    + c\left|\delta\right| \left|D(1)\right|\left|\psi(1)\right| + \left|\delta\right| \left\|\psi\right\|\left\|D\right\| +\left|\delta\right| \left\|\psi_{x}\right\|\left\|D_{x}\right\|\\
    + c\left|\theta(1)\right|\left|\psi(1)\right| + \left\|\psi\right\|\left\|\theta\right\| 
    + \left\|\psi_{x}\right\|\left\|\theta_{x}\right\|.
\end{align}
Now, we are ready to present the proofs of the parts separately.\\

\noindent \textit{Proof of the statement (P1)}: In the absence of uncertainty and attack, i.e. $\theta(x,t)=0,\delta(t)=0$, \eqref{lyap8} becomes
\begin{equation}
    \dot{W} \leqslant-c^2\left|\psi(1)\right|^2 - c\left\|\psi_{x}\right\|^2 - c\left\|\psi\right\|^2
   = -2cW, \label{lyap9} 
\end{equation}
the solution of which can be written using comparison principle as $W(t) \leqslant e^{-2ct}W(0)$. Considering the bounds in assumption \ref{as-init}, we can write $W(0) \leqslant(\frac{c}{2}+\frac{\overline{\psi}_0+\overline{\psi}_{x0}}{2\underline{\psi}_{10}}) \left|\psi(1,0)\right|^2$. Furthermore, by the choice of \eqref{lyap1}, we have $\frac{c}{2}\left|\psi(1,t)\right|^2 \leqslant W(t)$. Considering these two facts, and $r(t)=\tilde{u}(1,t)=\psi(1,t)$, we arrive at \eqref{part1}.\\

\noindent \textit{Proof of the statement (P2)}: In the presence of uncertainty and attack, i.e. $\theta(x,t)\neq 0,\delta(t) \neq 0$, we apply the upper bounds from Lemma \ref{D-bound} and re-write \eqref{lyap8} as
\begin{align}
    \dot{W} \leqslant-c^2\left|\psi(1)\right|^2 - c\left\|\psi_{x}\right\|^2 - c\left\|\psi\right\|^2 
    + c\left|\delta\right| \left|D_a(1)\right|\left|\psi(1)\right| 
    + \Big(1+\frac{c+\alpha}{2}\Big)\left|\delta\right| \left\|\psi\right\|\left\|D_a\right\|\\ +\Big(1+\frac{c+\alpha}{16}\Big)\left|\delta\right| \left\|\psi_{x}\right\|\left\|D_{ax}\right\|
    + c\left|\eta(1)\right|\left|\psi(1)\right| 
    + \Big(1+\frac{c+\alpha}{2}\Big)\left\|\psi\right\|\left\|\eta\right\| 
    + \Big(1+\frac{c+\alpha}{16}\Big)\left\|\psi_{x}\right\|\left\|\eta_{x}\right\|. \label{lyap10} 
\end{align}
Next, we apply Young's inequality on the last six terms of right hand side of \eqref{lyap10} and get
\begin{align}
    \dot{W}& \leqslant-c^2\left|\psi(1)\right|^2 - c\left\|\psi_{x}\right\|^2 - c\left\|\psi\right\|^2 
    + \frac{\lambda}{2}\left|\psi(1)\right|^2 +  \frac{c^2}{2\lambda}\left|D_a(1)\right|^2 \left|\delta\right|^2
     + \frac{(1+\frac{c+\alpha}{2})^2}{2\lambda}) \left\|D_a\right\|^2\left|\delta\right|^2 \\ 
   & + \frac{\lambda}{2}\left\|\psi\right\|^2 +\frac{\lambda}{2}\left\|\psi_{x}\right\|^2 + 
    \frac{(1+\frac{c+\alpha}{16})^2}{2\lambda} \left\|D_{ax}\right\|^2\left|\delta\right|^2 
    + \frac{\lambda}{2}\left|\psi(1)\right|^2 + \frac{c^2}{2\lambda}\left|\eta(1)\right|^2
    + \frac{\lambda}{2}\left\|\psi\right\|^2 
   \\& + \frac{(1+\frac{c+\alpha}{2})^2}{2\lambda}\left\|\eta\right\|^2 + \frac{\lambda}{2}\left\|\psi_{x}\right\|^2 + \frac{(1+\frac{c+\alpha}{16})^2}{2\lambda}\left\|\eta_{x}\right\|^2.\\
   \Rightarrow \dot{W} &\leqslant-2(c-\lambda)W + M_2 (\left\|D\right\|_{\mathcal{H}}^2 \left|\delta\right|^2 + \left\|\eta\right\|_{\mathcal{H}}^2). \label{lyap11} 
\end{align}
where $\lambda$ is a constant and $M_2$ is given in \eqref{nn1}. Applying comparison principle, the solution of the differential inequality \eqref{lyap11} can be written as
\begin{equation}
    W(t) \leqslant e^{-2(c-\lambda)t}W(0) + \frac{M_2}{2(c-\lambda)}\sup_{t\geqslant 0} \Big(\left\|D\right\|_{\mathcal{H}}^2 \left|\delta\right|^2 + \left\|\eta\right\|_{\mathcal{H}}^2\Big). \label{lyap12} 
\end{equation}
Again, considering the facts $W(0) \leqslant(\frac{c}{2}+\frac{\overline{\psi}_0+\overline{\psi}_{x0}}{2\underline{\psi}_{10}}) \left|\psi(1,0)\right|^2$, $\frac{c}{2}\left|\psi(1,t)\right|^2 \leqslant W(t)$, and $r(t)=\tilde{u}(1,t)=\psi(1,t)$, we arrive at \eqref{part2}.\\

\noindent \textit{Proof of the statement (P3)}: Consider \eqref{lyap10} under the condition $\eta(x,t) \neq 0, \delta = 0$. The seventh term of right hand side of \eqref{lyap10} can be written as:
\begin{align}
c\left|\eta(1)\right|\left|\psi(1)\right| = \alpha_1c\left|\eta(1)\right|\left|\psi(1)\right| + (1-\alpha_1)c\left|\eta(1)\right|\left|\psi(1)\right|,\label{lyap13} 
\end{align}
where $\alpha_1 \in \mathbb{R}$ is a constant. Furthermore, we apply Young's inequality on the first term of the right hand side of \eqref{lyap13} and get
\begin{align}
c\left|\eta(1)\right|\left|\psi(1)\right| \leqslant\alpha_1c\frac{\lambda_1}{2}\left|\eta(1)\right|^2 + \alpha_1c\frac{1}{2\lambda_1}\left|\psi(1)\right|^2 + (1-\alpha_1)c\left|\eta(1)\right|\left|\psi(1)\right|,\label{lyap14} 
\end{align}
with $\lambda_1 > 0$. In a similar manner, we can write the upper bounds of the eighth and ninth term of right hand side of \eqref{lyap10} as
\begin{align}
\Big(1+\frac{c+\alpha}{2}\Big)\left\|\psi\right\|\left\|\eta\right\| \leqslant\Big(1+\frac{c+\alpha}{2}\Big) \Big(\frac{\alpha_2\lambda_2}{2}\left\|\psi\right\|^2 + \frac{\alpha_2}{2\lambda_2}\left\|\eta\right\|^2 + (1-\alpha_2)\left\|\psi\right\|\left\|\eta\right\|\Big),\label{lyap15} 
\end{align}
\begin{align}
\Big(1+\frac{c+\alpha}{16}\Big)\left\|\psi_x\right\|\left\|\eta_x\right\| \leqslant\Big(1+\frac{c+\alpha}{16}\Big) \Big(\frac{\alpha_3\lambda_3}{2}\left\|\psi_x\right\|^2 + \frac{\alpha_3}{2\lambda_3}\left\|\eta_x\right\|^2 + (1-\alpha_3)\left\|\psi_x\right\|\left\|\eta_x\right\|\Big),\label{lyap16} 
\end{align}
with $\alpha_2,\alpha_3 \in \mathbb{R}$ are constants and $\lambda_2,\lambda_3>0$. Next, applying the upper bounds \eqref{lyap14}-\eqref{lyap16}, we can write the upper bound $\dot{W} \leqslant B_1$ where
\medmuskip=-2mu
\thinmuskip=-2mu
\thickmuskip=-2mu
\begin{align}
    B_1 = \left|\psi(1)\right|^2 \Big( -c^2 + \frac{c\alpha_1\lambda_1}{2}\Big)
    + \left\|\psi\right\|^2 \Big( -c + (1+\frac{c+\alpha}{2})\frac{\alpha_2\lambda_2}{2}\Big)
    + \left\|\psi_x\right\|^2 \Big( -c + (1+\frac{c+\alpha}{16})\frac{\alpha_3\lambda_3}{2}\Big)
    + \left|\eta(1)\right|^2 \Big(\frac{c\alpha_1}{2\lambda_1}\Big)\\
    + \left\|\eta\right\|^2 \Big( (1+\frac{c+\alpha}{2})\frac{\alpha_2}{2\lambda_2}\Big)
    + \left\|\eta_x\right\|^2 \Big( (1+\frac{c+\alpha}{16})\frac{\alpha_3}{2\lambda_3}\Big)
    + \left|\psi(1)\right|\left|\eta(1)\right| (c(1-\alpha_1))
    + \left\|\psi\right\|\left\|\eta\right\| ((1+\frac{c+\alpha}{2})(1-\alpha_2))\\
    + \left\|\psi_x\right\|\left\|\eta_x\right\|((1+\frac{c+\alpha}{16})(1-\alpha_3)). \label{lyap17} 
\end{align}
\thinmuskip=3mu
\medmuskip=4mu plus 2mu minus 4mu
\thickmuskip=5mu plus 5mu
Next, we define the following metric
\begin{equation}
    \mathcal{K}_1 = \int_0^\infty \Big( \left|\psi(1)\right|^2 - \beta_1^2\left|\eta(1)\right|^2 - \beta_1^2\left\|\eta\right\|^2 - \beta_1^2\left\|\eta_x\right\|^2 \Big)d\tau. \label{k1} 
\end{equation}
Note that the condition $\mathcal{K}_1 - \epsilon \leqslant 0$ is same as the condition \eqref{part3}. Furthermore, we know from \eqref{lyap17} that
\begin{equation}
    W(\infty) -W(0) - \int_0^\infty B_1 d\tau \leqslant 0. \label{lyap18} 
\end{equation}
Hence, proving 
\begin{equation}
   \mathcal{K}_1 - \epsilon \leqslant W(\infty) -W(0) - \int_0^\infty B_1 d\tau, \label{lyap19} 
\end{equation}
would be sufficient to achieve $\mathcal{K}_1 - \epsilon \leqslant 0$. Furthermore, noting $W(\infty) \geqslant 0$ and $W(0)-\epsilon \leqslant 0$ with the choice $\epsilon = \frac{c}{2}\overline{\psi}_{10}^2 + \overline{\psi}_0^2 + \overline{\psi}_{x0}^2$, we find  that
\begin{equation}
   \mathcal{K}_1 + \int_0^\infty B_1 d\tau \leqslant 0, \label{lyap20} 
\end{equation}
is a sufficient condition for proving \eqref{lyap19}. Using \eqref{k1}, we can write the condition \eqref{lyap20} as
\begin{equation}
   \int_0^\infty \Big( \left|\psi(1)\right|^2 - \beta_1^2\left|\eta(1)\right|^2 - \beta_1^2\left\|\eta\right\|^2 - \beta_1^2\left\|\eta_x\right\|^2 + B_1 \Big)d\tau \leqslant 0. \label{lyap21} 
\end{equation}
Finally, we find that proving 
\begin{equation}
    \left|\psi(1)\right|^2 - \beta_1^2\left|\eta(1)\right|^2 - \beta_1^2\left\|\eta\right\|^2 - \beta_1^2\left\|\eta_x\right\|^2 + B_1  \leqslant 0, \label{lyap22} 
\end{equation}
would be sufficient for proving \eqref{lyap21}. In summary, proving \eqref{lyap22} would automatically satisfy the original robustness condition \eqref{part3}. The condition \eqref{lyap22} can be expressed as $\mathbb{X}_1^T \mathbb{A} \mathbb{X}_1 \leqslant 0$ with $\mathbb{X}_1 = [\left|\psi(1)\right|,\left\|\psi\right\|,\left\|\psi_x\right\|,\left|\eta(1)\right|,\left\|\eta\right\|,\left\|\eta_x\right\|]^T$ and where $\mathbb{A}$ is defined in \eqref{cond2a}. Consequently, satisfying the LMI $\mathbb{A} \leqslant 0$ would be sufficient to satisfy the robustness condition \eqref{part3}.\\

\noindent \textit{Proof of the statement (P4)}: Consider \eqref{lyap10} under the condition $\eta(x,t) = 0, \delta(t) \neq 0$. Following the same approach as in (P3), we can write the upper bound $\dot{W} \leqslant B_2$ where
\medmuskip=-2mu
\thinmuskip=-2mu
\thickmuskip=-2mu
\begin{align}\label{lyap23} \nonumber
    B_2 = &\left|\psi(1)\right|^2 \Bigg( -c^2 + \frac{c\alpha_4\lambda_4}{2}\Bigg)
    + \left\|\psi\right\|^2 \Bigg( -c + \left(1+\frac{c+\alpha}{2}\right)\frac{\alpha_5\lambda_5}{2}\Bigg)
    + \left\|\psi_x\right\|^2 \Big( -c + \left(1+\frac{c+\alpha}{16}\right)\frac{\alpha_6\lambda_6}{2}\Bigg)
    + \left|\delta\right|^2\left|D_a(1)\right|^2 \Bigg(\frac{c\alpha_4}{2\lambda_4}\Bigg)\\\nonumber
   & + \left|\delta\right|^2\left\|D_a\right\|^2 \Bigg( \left(1+\frac{c+\alpha}{2}\right)\frac{\alpha_5}{2\lambda_5}\Bigg)
    + \left|\delta\right|^2\left\|D_{ax}\right\|^2 \Bigg(\left (1+\frac{c+\alpha}{16}\right)\frac{\alpha_6}{2\lambda_6}\Bigg)
    + \left|\psi(1)\right|\left|\delta\right|\left|D_a(1)\right| c(1-\alpha_4)
   \\ &+ \left\|\psi\right\|\left|\delta\right|\left\|D_a\right\|\left(1+\frac{c+\alpha}{2}\right)(1-\alpha_5)
    + \left\|\psi_x\right\|\left|\delta\right|\left\|D_{ax}\right\|\left(1+\frac{c+\alpha}{16}\right)(1-\alpha_6), 
\end{align}
\thinmuskip=3mu
\medmuskip=4mu plus 2mu minus 4mu
\thickmuskip=5mu plus 5mu
with $\alpha_4,\alpha_5,\alpha_6 \in \mathbb{R}$ are constants and $\lambda_4,\lambda_5,\lambda_6 >0$. Next, we define the following metric
\begin{align}
    \mathcal{K}_2 = \int_0^\infty \Big( \left|\psi(1)\right|^2 - \beta_2^2\left|\delta\right|^2\left|D_a(1)\right|^2 - \beta_2^2\left|\delta\right|^2\left\|D_a\right\|^2 - \beta_2^2\left|\delta\right|^2\left\|D_{ax}\right\|^2 \Big)d\tau. \label{k2} 
\end{align}
Note that the condition $\mathcal{K}_2 + \epsilon \geqslant 0$ is same as the condition \eqref{part4}. Next, following a similar approach as in (P3), we find  that
\begin{equation}
   \mathcal{K}_2 - \int_0^\infty B_2 d\tau \geqslant  0, \label{lyap24} 
\end{equation}
is a sufficient condition for proving \eqref{part4} and consequently, proving 
\medmuskip=-2mu
\thinmuskip=-2mu
\thickmuskip=-2mu
\begin{equation}
    \left|\psi(1)\right|^2 - \beta_2^2\left|\delta\right|^2\left|D_a(1)\right|^2 - \beta_2^2\left|\delta\right|^2\left\|D_a\right\|^2 - \beta_2^2\left|\delta\right|^2\left\|D_{ax}\right\|^2 - B_2  \geqslant  0, \label{lyap25} 
\end{equation}
\thinmuskip=3mu
\medmuskip=4mu plus 2mu minus 4mu
\thickmuskip=5mu plus 5mu
would be sufficient for satisfying \eqref{part4}. The condition \eqref{lyap25} can be expressed as $\mathbb{X}_2^T \mathbb{B} \mathbb{X}_2 \geqslant  0$ with $\mathbb{X}_2 = [\left|\psi(1)\right|,\left\|\psi\right\|,\left\|\psi_x\right\|,\left|\delta\right|\left|D_a(1)\right|,\left|\delta\right|\left\|D_a\right\|,\left|\delta\right|\left\|D_{ax}\right\|]^T$ and where $\mathbb{B}$ is defined in \eqref{cond3a}. Consequently, satisfying the LMI $\mathbb{B} \geqslant 0$ would be sufficient to satisfy the sensitivity condition \eqref{part4}.
\end{proof}

\begin{remm}\label{design}
For design purposes, one can solve the Linear Matrix Inequalities (LMIs) \eqref{cond2} along with the condition $c > \lambda>0$. The parameters $\alpha_i \in \mathbb{R}$ and $\lambda_i >0$ can be utilized as tuning parameters. 
To account for the presence of uncertainty in the system, we modify the attack detection logic \eqref{residual} as: $r(t) \leqslant \epsilon_{th} \Rightarrow \text{No attack},r(t) > \epsilon_{th} \Rightarrow \text{Attack}$.
where $\epsilon_{th}$ is a pre-defined threshold. There are multiple ways one can compute such threshold, for example, $\epsilon_{th} = \max_{\delta=0,\eta\neq 0}r(t)$.
\end{remm}

\section{SIMULATION CASE STUDIES}
In this section, we perform simulation studies to illustrate the performance of attack detection algorithms. We consider battery under cyber-attack. Cyber-attacks on electric vehicle or power-grid batteries can have severe implications \cite{dey2020cybersecurity}. Specifically, we focus on cyber-attacks that attempt manipulate battery current to cause over-temperature conditions. We adopt the battery model from \cite{muratori2012spatially}. Under the assumption sufficient boundary insulation and applying transformation of time and spatial variables, the battery model can be written as:
\begin{align}
    &T_t(x,t) = T_{xx}(x,t) + D(x)q(t) + D(x)\delta(t), \label{pde-sim}\\
    &T_x(0,t) = 0, \quad T_x(1,t) = 0,  \label{pde-bc-sim}
\end{align}
where $T(x,t)$ is the distributed battery temperature; $D(x) = K, \forall x\in [0,1]$ where the constant $K$ depends on battery dimension and thermal properties; the measured output is the boundary temperature $y(t)=T(1,t)$; $q(t)$ is the nominal heat generated within the battery due to current flow and $\delta(t)$ is the attack signal that modifies the heat generation by corrupting battery current. The distributed state response under constant current discharge and nominal conditions (i.e. no attack and no uncertainty) are shown in Fig. \ref{Fig1}. Next, we present the following case studies to illustrate the ideas presented in the previous sections.

\begin{figure}[!h]
\centering
\includegraphics[trim = 0mm 0mm 0mm 0mm, clip, scale=0.5, width=0.5\linewidth]{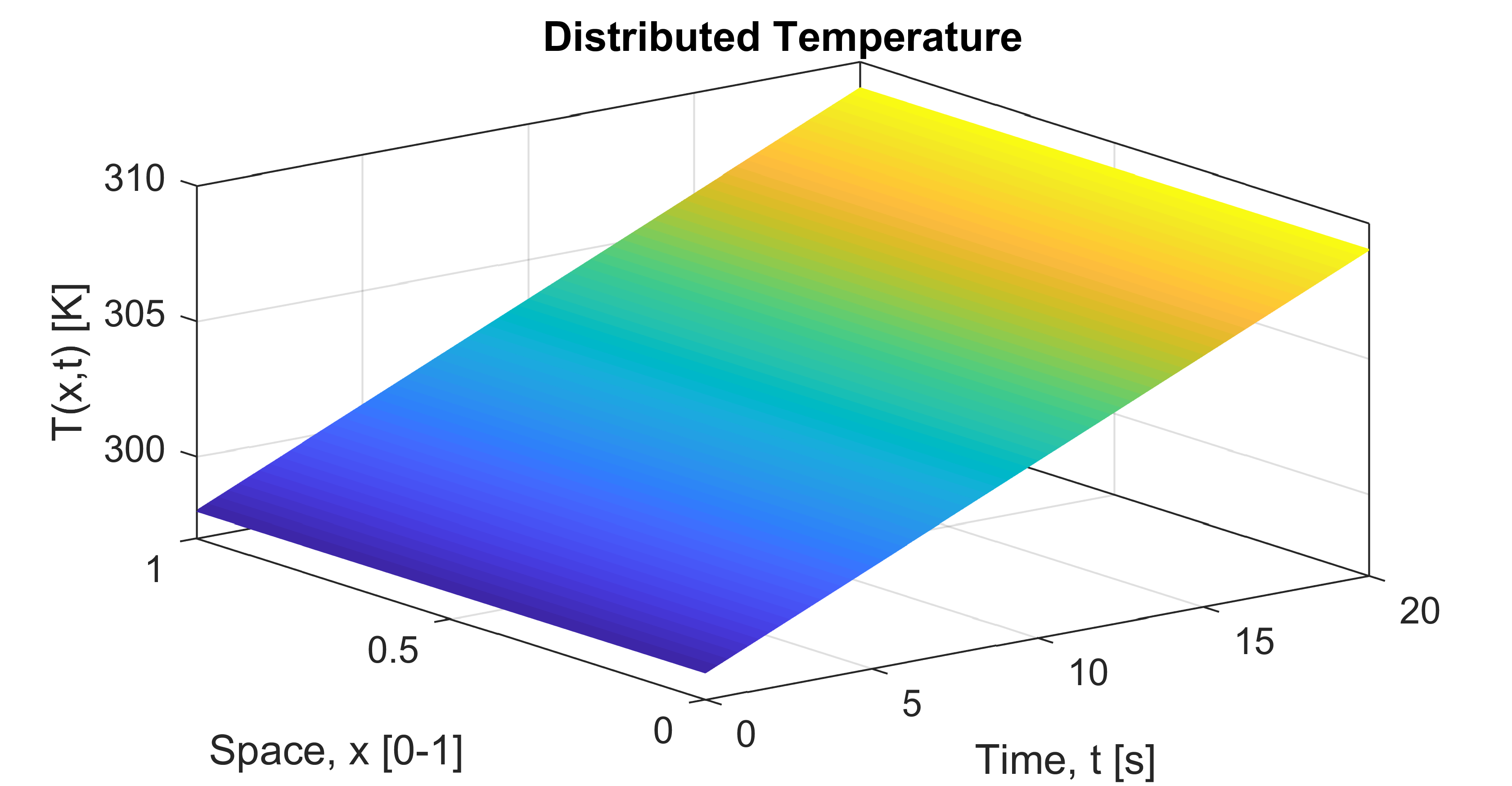}
\caption{Distributed battery temperature response under constant current scenario.}
\label{Fig1}
\end{figure}

In the first study, we illustrate the existence of a stealthy cyber-attack in the sense of Definition 1. We compare two cases: (i) \textit{Case 1:} where the system is subjected to initial condition $T(x,0)=290 K$ and a short pulse type stealthy cyber-attack $\delta(t)$ applied at $t=0$. (ii) \textit{Case 2:} where the system is subjected to initial condition $T(x,0)=298 K$ and no cyber-attack has been applied. Note that the initial conditions fall within reasonable range for both cases and the response under attack resembles the response under nominal scenario. The responses for these cases are shown in Fig. \ref{Fig2}. Hence, by Definition 1, the applied cyber-attack possesses stealthiness.

\begin{figure}[!h]
\centering
\includegraphics[trim = 0mm 0mm 0mm 0mm, clip, scale=0.5, width=0.5\linewidth]{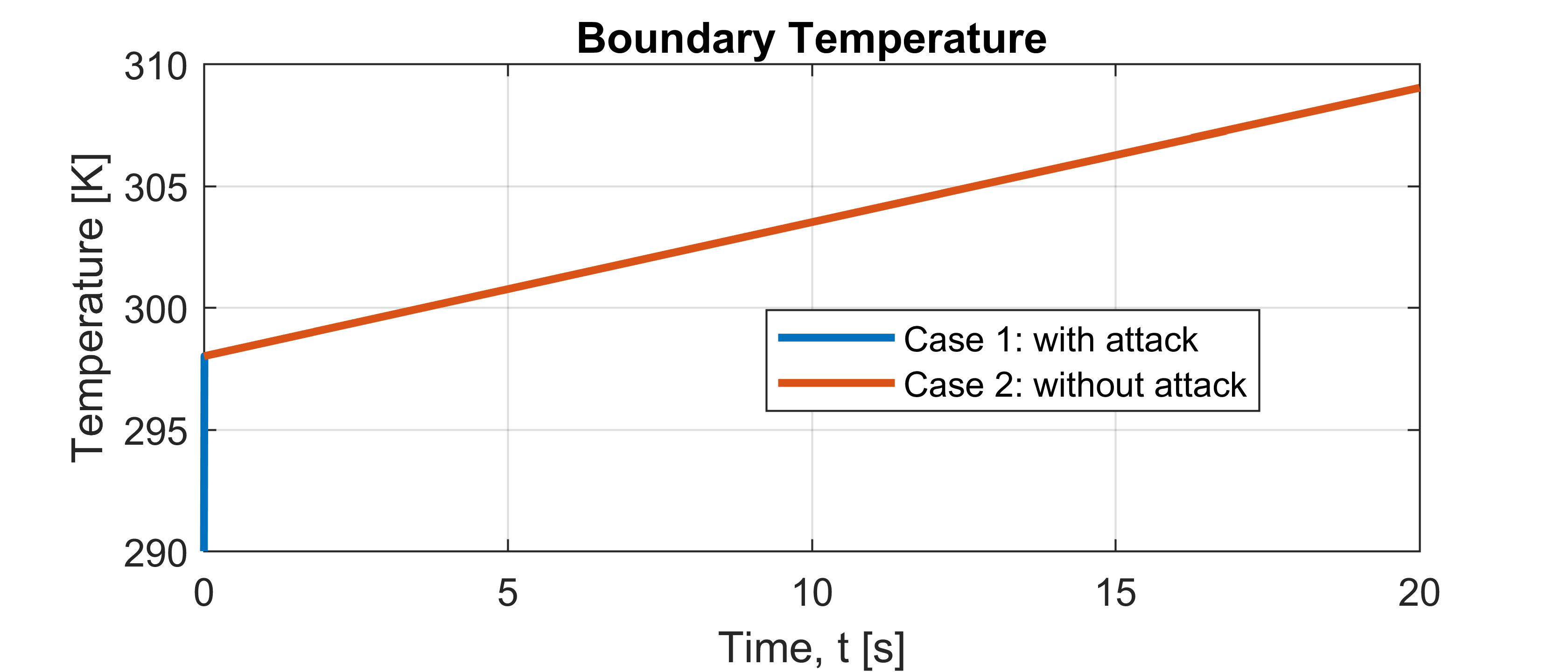}
\caption{Comparison of boundary temperature response with and without attack.}
\label{Fig2}
\end{figure}

In the next study, we illustrate the performance of the attack detection algorithm discussed in Section IV. We illustrate the results in $\psi$-domain. The states of the PDE observer were initialized incorrectly to verify the convergence. We consider three scenarios: (i) without any uncertainty and attack, (ii) with uncertainty and no attack, and (iii) with uncertainty and attack. The uncertainty is chosen as $\mathcal{N}(0,1)$ distributed spatially. The attack signal is chosen as $\delta(t)=0.0015(1-e^{-0.0003(t-T_a)})$ with $T_a = 10 s$ being the attack injection time. The threshold $\epsilon_{th}$ has been chosen according to Remark \ref{design}. The responses of the residual signal under these three scenarios are shown in Fig. \ref{Fig3}. As proved in Theorem 3, the residual signal (i) exponentially converges to zero in the absence of any attack or uncertainty, (ii) exhibits Input-to-State (ISS) stability in the presence of attack and/or uncertainty, (iii) remains bounded within the threshold in the absence of an attack, and (iv) crosses the threshold $2$ s after the attack occurrence thereby detecting the attack. In summary, the statements of Theorem 3 have been verified by this case study.

\begin{figure}[!h]
\centering
\includegraphics[trim = 0mm 0mm 0mm 0mm, clip, scale=0.5, width=0.5\linewidth]{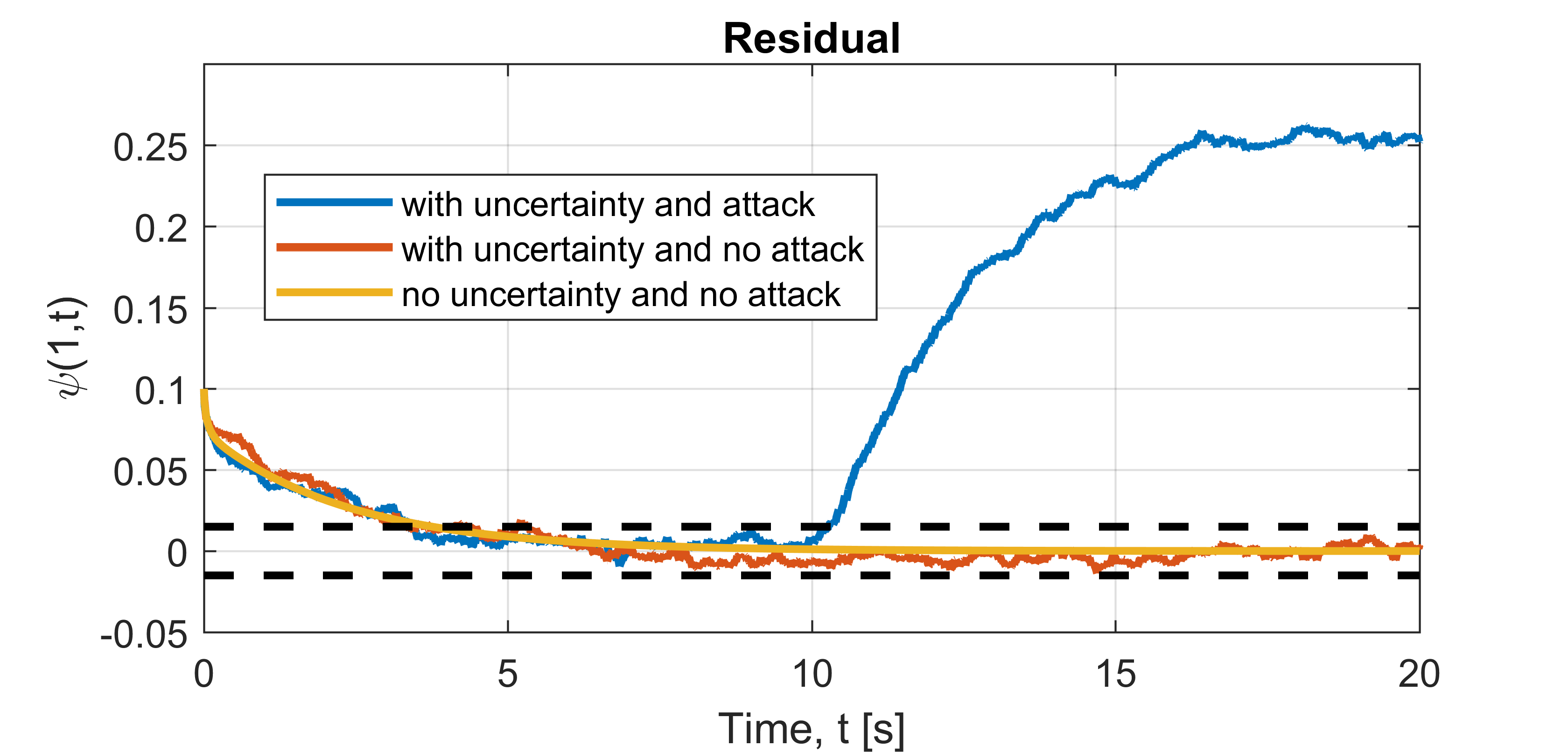}
\caption{Residual responses under uncertainty and attack.}
\label{Fig3}
\end{figure}



\section{CONCLUSION}
In this paper, we have explored security of DPCPSs modelled by linear parabolic PDEs with boundary measurements. The focus is on cyber-attacks that affect the actuation channel. First, we consider the detectability aspects of such cyber-attacks. Subsequently, we analyze existence of stealthy attacks under a special class of algorithms utilizing system model and measurements. Next, we develop a design framework that explicitly considers stability, and the trade-off between robustness and attack sensitivity in its design phase. As a future work, we plan to extend the framework to n-dimensional PDE systems.

\bibliography{ref1}  

\end{document}